\documentclass[runningheads, 11pt]{llncs}
\pdfoutput=1 
\usepackage{amsmath}
\usepackage{enumitem}
\usepackage{amssymb}
\usepackage{commath}
\usepackage[hidelinks]{hyperref}
\usepackage{algorithm}
\usepackage{algorithmicx}
\usepackage{algpseudocode}%

\def \vdf	{\textsc{VDF}}
\def \poly	{\mathtt{poly}}
\def \negl	{\mathtt{negl}}
\def \RO	{\mathsf{H}}
\def \adv       {\mathcal{A}}
\def \adf       {\mathcal{B}}
\def \pp   	{\mathbf{pp}}
\def \setup	{\mathsf{Setup}}
\def \eval	{\mathsf{Eval}}
\def \Verify	{\mathsf{Verify}}
\def \st	{\mathtt{state}}
\def \mod       {\; \mathbf{mod} \;}
\def \multgroup#1{(\mathbb{Z}/#1\mathbb{Z})^\times}
\def \vrf       {\mathcal{V}}
\def \prv       {\mathcal{P}}
\def \X         {\mathcal{X}}
\def \Y         {\mathcal{Y}}
\def \O     	{\mathcal{O}}
\def \QR 	{\mathbb{QR}}
\def \F 	{\mathbb{F}}
\def \G		{\mathbb{G}}
\def \PRIMES 	{\mathbb{P}}
\def \Z 	{\mathbb{Z}}

\newtheorem{assumption}{Assumption.}
\begin{document}
\title{Single Squaring Verifiable Delay Function from Time-lock Puzzle in the Group of Known Order}

\author{Souvik Sur \orcidID{0000-0003-1109-8595}}

\institute{\email{souviksur@gmail.com}}

\maketitle              

\begin{abstract}\label{abstract}
\hyphenpenalty=8000\exhyphenpenalty=8000\tolerance=500\pretolerance=500%
A Verifiable Delay Function ($\vdf$) is a function that takes
a specified sequential time $T$ to be evaluated, but can be verified in $\O(\log T)$-time. 
For meaningful security, $T$ can be at most subexponential in the security
parameter $\lambda$ but has no lower bound.
$\vdf$s are useful in several applications ranging from randomness beacons 
to sustainable blockchains but are really rare in practice.
The verification in all of these $\vdf$s requires $\Omega(\lambda,\log{T})$
sequential time.
%

This paper derives a verifiable delay function that is verifiable in $\O(1)$-sequential
time. The key observation is that the prior works use subexponentially-hard 
algebraic assumptions for their sequentiality.
On the contrary, we derive our $\vdf$ from a polynomially-hard sequential
assumption namely the time-lock puzzle over the group of known order. In particular, we
show that time-lock puzzle can be sequentially-hard even when the order of the group is
known but the delay parameter is polynomially-bounded
i.e., $T\le \poly(\lambda)$. 
As an add-on advantage, our $\vdf$ requires only one sequential squaring to verify.
Thus, in our $\vdf$, the sequential effort required for 
verification is fixed and independent of the delay parameter $T$.

\end{abstract}
\section{Introduction}\label{introduction}

 The notion of verifiable delay functions was
 introduced in~\cite{Dan2018VDF}. A verifiable delay function is a 
 function with the domain $\mathcal{X}$ and the range $\mathcal{Y}$, that takes a specified 
 number of sequential steps $T$ to be evaluated (irrespective of the amount of parallelism) and 
 can be verified efficiently (even without parallelism) and publicly. In order to avoid exponential (processors)
 adversary $T=2^{o(\lambda)}$ at most.

 Along with its introduction, a candidate $\vdf$ using injective rational maps has been proposed in~\cite{Dan2018VDF}.
 However, its prover needs a certain amount of parallelism to evaluate.
 Wesolowski~\cite{Wesolowski2019Efficient} and Pietrzak~\cite{Pietrzak2019Simple} come up with two 
 $\vdf$s separately, although based on the same hardness assumption of time-lock puzzle
 in the group of unknown order~\cite{Rivest1996Time}. Feo et al.~\cite{Feo2019Isogenie} propose a $\vdf$ based 
 on super-singular elliptic curves defined over finite fields. The verifiers in all of
 these $\vdf$s incur $\Omega{\lambda,\log{T}}$ sequential time in verification.


Therefore, an intriguing question would be is it possible to 
design a $\vdf$ that is verifiable in \emph{sequential} time independent of $\lambda$ and $T$. 

\subsection{Technical Overview}\label{contributions}

We show that it is possible to have $\vdf$ 
from polynomially-hard sequential function i.e., the delay parameter $T\le \poly(\lambda)$.
The prior works do not enjoy faster verification because of their hardness assumptions 
used in the sequentiality are essentially sub-exponentially hard.

First we find that time-lock puzzle over the group of \emph{known} orders can also be sequentially but
polynomially hard. Then, we propose a $\vdf$ that needs only a single sequential squaring for verification.
Thus the verification effort (sequential time) is independent of the security parameter $\lambda$ and 
the delay parameter $T$.

%

Briefly, our scheme works as follows. Depending upon $\lambda$ and $T$, it chooses 
a random oracle $\RO:\X\rightarrow \QR^*_q$ such that $\QR^*_q$ is the semigroup of
squares except $1$ in a finite field $\F_q$ of order $q$.
Then, the prover is asked to find the square root of an arbitrary element in $\F^*_q$. 
Given the input $x\in\mathcal{X}$, the prover needs to compute
$\RO(x)^{\frac{q+1}{4}}\in\F^*_q$. 
During verification, the verifier accepts if and only if 
$y^2=\RO(x)$.

We show that our construction is correct, sound and sequential. 
Further, our $\vdf$ needs no proof, thus is only a one round protocol.
We summarize a bunch of other advantages and also a disadvantage 
of our $\vdf$ over the existing ones in Sect.~\ref{compare}.
   

\subsection{Organization of the Paper}
 This paper is organized as follows.
Section~\ref{literature} discusses a few existing schemes known to be  
$\vdf$. In Section~\ref{preliminaries}, we present a succinct review of
$\vdf$. 
We propose our single squaring verifiable delay function in Section~\ref{LPoSW}.
In Section~\ref{properties} we establish the essential 
properties correctness, sequentiality and soundness of the $\vdf$.
Finally, Section~\ref{dis} concludes the paper. 

\section{Related Work}\label{literature}
In this section, we mention some well-known schemes qualified as $\vdf$s, and summarize
their features in Table~\ref{tab : VDF} in Sect.~\ref{compare}. We categorize them by
the sequentiality assumptions.


\paragraph{Injective Rational Maps}
In 2018, Boneh et al.~\cite{Dan2018VDF} propose a $\vdf$ based on injective rational
maps of degree $T$, where the fastest possible inversion is to compute the polynomial
GCD of degree-$T$ polynomials. They conjecture that 
it achieves $(T^2,o(T))$ sequentiality using permutation polynomials as the candidate map.
However, it is a weak $\vdf$ as it needs $\mathcal{O}(T)$ processors to evaluate the output in time $T$.

\paragraph{Time-lock Puzzle}\label{vdf}
Rivest, Shamir and Wagner introduced the time-lock puzzle stating
that it needs at least $T$ number of sequential squaring to compute 
$y=g^{2^T}\mod{\Delta}$ when the factorization of $\Delta$ is unknown~\cite{Rivest1996Time}.
Therefore they proposed this encryption that can be decrypted 
only sequentially. Starting with $\Delta=pq$ such that $p,q$ are large primes,
the key $y$ is enumerated as $y=g^{2^T}\mod{\Delta}$. Then the verifier,
uses the value of $\phi(\Delta)$ to reduce the exponent to
$e=2^T\mod{\phi(\Delta)}$ and finds out $y= g^e\mod{\Delta}$.
On the contrary, without the knowledge of $\phi(\Delta)$, the only option available to the prover
is to raise $g$ to the power $2^T$ sequentially. 
As the verification stands upon a secret, the knowledge of $\phi(\Delta)$, 
it is not a $\vdf$ as verification should depend only on public parameters.

Pietrzak~\cite{Pietrzak2019Simple} and Wesolowski~\cite{Wesolowski2019Efficient} circumvent 
this issue independently. We describe both the $\vdf$s in the generic group
$\mathbb{G}$ as done in~\cite{VDF2018Survey}.
These protocols can be instantiated over two different groups --
the RSA group $\multgroup{N}$ and the class group of imaginary quadratic
number field $Cl(d)$. Both the protocols use a common random oracle
$\RO_\mathbb{G}:\{0,1\}^*\rightarrow \mathbb{G}$ to map the input statement $x$ to the
generic group $\mathbb{G}$. We assume $g:=\RO_\mathbb{G}(x)$. The sequentiality of both
the protocols are shown by the reductions from the time-lock puzzle. For soundness they
use different hardness assumptions in the group $\G$.

\paragraph{Pietrzak's $\vdf$}
Pietrzak's $\vdf$ exploits the identity $(g^{r\cdot {2^{T/2}}} \times g^{2^T})=
(g^r\times g^{2^{T/2}})^{2^{T/2}}$ for any random integer $r$. The prover $\prv$ needs to compute
$y=g^{2^T}$ and the verifier $\vrf$ checks it by sampling a random $r$ in that identity.
In the non-interactive setting, the integer $r \in \Z_{2^\lambda}$ is sampled by another
random oracle $\RO:\Z \times \G \times \G \times \G  \rightarrow \Z_{2^\lambda}$.
The above-mentioned identity relates all $v=g^{r\cdot {2^{T/2}}} \times g^{2^T}$ and 
$u=g^r\times g^{2^{T/2}}$ for any $r\in\Z^+$. So, $\vrf$ engages $\prv$ in proving the
identity for $T/2, T/4, \ldots, 1$  in $\log T$ number of rounds. In each round $i$, a new
pair of $(u_i,v_i)$ is computed depending upon the integer $r_i$ as follows.

First, the prover $\prv$ initializes $u_1=g$ and $v_1=y$.
Then $\prv$ is asked to compute the output $y=g^{2^T}$ and the proof 
$\pi=\{z_1, z_2, \ldots, z_{\log{T}}\}$ such that $u_{i+1}=u_i^{r_i}\cdot z_i$ and $v_{i+1}=z_i^{r_i}\cdot v_i$ where 
$z_i=u_i^{2^{T/2^i}}$ and $r_i = \RO(u_i, T/2^{i-1},v_i,z_i)$. 
The effort to generate the proof $\pi$ is in $\O(\sqrt{T} \log{T})$.
The verifier $\vrf$ reconstructs $u_{i+1}=u_i^{r_i}\cdot z_i$ and $v_{i+1}=z_i^{r_i}\cdot v_i$
by fixing $u_1$,  $v_1$ and $r_i$s as before. 
Finally, $\vrf$ accepts the proof if and only if $v_{\log{T}}=u_{\log{T}}^2$.
As $r_{i+1}$ depends on $r_i$ they need to be computed sequentially.
So, doing at most $\log r_i$ squaring in each round, the verifier performs
$\sum_{i=1}^{\log{T}}\log{r_i}=\O(\lambda\log{T})$ squaring in total.
Sect.~\ref{minseq} presents a detailed inspection on the lower bound of verification. 

The soundness of this protocol is based on the low order 
assumption in the group $\G$. It assumes that  
finding an element with the order $<2^\lambda$ is computationally hard. A crucial
feature that comes from this assumption is that the statistical soundness. This $\vdf$
stands sound even against computationally unbounded adversary. The probability that this
$\vdf$ accepts a false proof produced by even a computationally unbounded adversary is
negligible in the security parameter $\lambda$. 

\paragraph{Wesolowski's $\vdf$}

It has been designed using the identity $g^{2^T}=(g^q)^\ell \times g^r$ where
$2^T=q\ell +r$. Here $\ell$ is a prime chosen uniformly at random from the set of first
$2^{2\lambda}$ primes $\PRIMES_{2\lambda}$. 

This protocol asks the prover $\prv$ to compute the output $y=g^{2^T}$. Then, the
verifier $\vrf$ chooses the prime $\ell$ from the set $\PRIMES_{2\lambda}$. In the
non-interactive version, the prime 
$\ell=\RO_{\texttt{prime}}(\texttt{bin}(g)|||\texttt{bin}(y))$ is
sampled by the random oracle $\RO_{\texttt{prime}}: \G \times \G \rightarrow
\PRIMES_{2\lambda}$. Now, $\prv$ needs to compute the proof $\pi=g^q$ where
the quotient $q={\lfloor2^T/\ell\rfloor}$. The verifier $\vrf$ finds the remainder
$r=2^T \mod {\ell}$. Then $\vrf$ checks if $y=\pi^\ell \cdot g^r$. This holds true since
$2^T=q \ell + r$ so $g^{2^T}=(g^q)^\ell \times g^r$. 

It has been shown in~\cite{Wesolowski2019Efficient}, that $\prv$ needs $\O(T/\log {T})$
time to compute the proof $\pi$. On the other hand, the verifier $\vrf$ needs to compute
$\pi^\ell$ and $g^r$. As $r < \ell$, the verification needs at most $2\log \ell$ squaring.
Since, $\ell$ is a prime of size at most $2\lambda\log{2\lambda}$ (by prime number
theorem) the verification amounts to at most $2\log{\ell} = \O(\lambda\log{\lambda})$ squaring.
In Sect.~\ref{minseq}, we review the lower bound of its verification. 

For the soundness, Wesolowski's $\vdf$ relies on the adaptive root assumption in group
$\G$. It states that given a prime $\ell$ and an arbitrary element in the group $\G$ 
it is computationally hard to find the $\ell$-th root of that element.

\paragraph{Continuous $\vdf$}
In a beautiful adaption of Pietrzak's construction, Ephraim et al. introduce a unique
feature called continuous verifiability into the world of $\vdf$s. 
Any intermediate state at time $T'<T$ of their $\vdf$ can be continuously
verified in $\O(\lambda)$ time~\cite{Ephraim20VDF}.  

\paragraph{Isogenies over Super-Singular Curves}
Feo et al. presents two $\vdf$s based on isogenies over super-singular
elliptic curves~\cite{Feo2019Isogenie}. They start with five groups $\langle
\G_1,\G_2,\G_3,\G_4,\G_5\rangle$
of prime order $T$ with two non-degenerate bilinear pairing maps
$e_{12}: \G_1 \times \G_2 \rightarrow \G_5$ and $e_{34}: \G_3 \times \G_4 \rightarrow \G_5$.
Also there are two group isomorphisms
$\phi: \G_1 \rightarrow \G_3$ and $\overline{\phi}: \G_4 \rightarrow \G_2$. 
Given all the above descriptions as the public parameters along with a generator $P\in \G_1$,
the prover needs to find $\overline{\phi}(Q)$, where $Q\in \G_4$, using $T$ sequential steps.
The verifier checks if $e_{12}(P,\overline{\phi}(Q))=e_{34}(\phi(P),Q)$ in
$\poly(\log{T})$ time. It runs on super-singular curves over $\mathbb{F}_p$ and $\mathbb{F}_{p^2}$ 
as two candidate groups.

\section{Preliminaries}\label{preliminaries}
Now we mention the notations and terminology used in this paper.

\subsection{Notation}

We denote the security parameter with $\lambda\in\mathbb{Z}^+$.
The term $\poly(\lambda$) refers to some polynomial of $\lambda$, and
$\negl(\lambda$) represents some function $\lambda^{-\omega(1)}$.
If any randomized algorithm $\adv$ outputs $y$ on an input $x$, 
we write $y\xleftarrow{R}\adv(x)$. By $x\xleftarrow{\$}\mathcal{X}$,
we mean that $x$ is sampled uniformly at random from $\mathcal{X}$. 
For any set $\mathcal{X}$, $|\mathcal{X}|$ denotes 
the cardinality of the set $\mathcal{X}$, whereas
for everything else encoded as strings, 
$|x|=n$ denotes the bit-length of the string $x\in\{0,1\}^n$. 

We consider $\adv$ as efficient if it runs in 
probabilistic polynomial time (PPT). We assume (or believe) a problem to be hard 
if it is yet to have an efficient algorithm for that problem.
We say that an algorithm $\adv$ runs in parallel
time $T$ with $\Gamma$ processors if it can be implemented on a PRAM machine
with $\Gamma$ parallel processors running in time $T$.

\subsection{Square Roots in Finite Fields of Odd Characteristics}

In this section, we discuss the issues of finding square roots of an element in
a finite field. It will be required in the security proofs in Sect.~\ref{properties}. 

Any field $\F_q$ can be succinctly represented using the fact that $\F_q$ is always isomorphic
to the ring of polynomials $\F_p[x]/\langle f(x) \rangle$ where $p$ is a prime and $f(x)$
is an irreducible polynomial over $\F_p$ of degree $n$ such that $q=p^n$. Therefore, it
suffices to specify $p$ and $f(x)$ to describe $\F_q$. The multiplicative cyclic group
of $\F_q$ is denoted as $\F^*_q$ of order $q-1$.

Let $\QR_q$ be the set of squares in $\F^*_q$ i.e., $\QR_q=\{s^2 \mid \forall
s \in\F^*_q\}$ of order $(q-1)/2$.
We define the semigroup of non-trivial squares
as, $$\QR^*_q\stackrel{def}{=}\QR_q\setminus\{1\}.$$

\begin{lemma}\label{order}
In a finite field of odd order $q$,
\begin{enumerate}[label=\roman*.]
\item for all $a\in\F^*_q$, $a\in\QR_q$ if and only if $a^{\frac{q-1}{2}}=1$,
\item The order of $\QR_q$ is $|\QR_q|=\frac{q-1}{2}$. 
\end{enumerate} 
\end{lemma}

\begin{proof}
By Lagrange's theorem, all elements in $\F^*_q$ satifies the equation $x^{q-1}=1$. 
It means $x^{\frac{q-1}{2}}=\pm 1$ holds true for all elements in $\F^*_q$. 
Therefore, the polynomial $x^{q-1}-1$ has exactly $(q-1)$ roots over $\F^*_q$.

\begin{description}

\item [Part $i.$]
If $a\in\QR_q$ then for some $x\in\F^*_q$, $a=x^2$. Then, $a^{\frac{q-1}{2}}=x^{q-1}=1$. 
Conversely, if $a^{\frac{q-1}{2}}=1$ then $a^{\frac{1}{2}}$ must be a root of the
equation $x^{q-1}=1$. Then, by Lagrange's theorem, 
$a^\frac{1}{2}\in \F^*_q$, so $a\in\QR_q$. 

\item [Part $ii.$]
Since, the orders of all the elements divide the order of the group,
$|\QR_q|=k\cdot\frac{(q-1)}{2}$ for some $k\in\Z^+$. Since, $x^{\frac{q-1}{2}}=-1$ for some
$x\in\F^*_q$, we have $\QR_q\subset\F^*_q$, so
$|\QR_q|< |\F^*_q|$. Both together imply that $|\QR_q|=\frac{(q-1)}{2}$.
\end{description}
\qed
\end{proof}

\begin{lemma}\label{member}
In a finite field of odd order $q$, $-1\notin\QR_q$ if and only if $q=3\mod{4}$. 
\end{lemma}

\begin{proof}
If $q=3\mod{4}$ then $\frac{q-1}{2}=1 \mod{2}$. Therefore,
$-1^{\frac{q-1}{2}}=-1\notin\QR_q$. Conversely, if $-1\notin\QR_q$ 
then $\frac{q-1}{2}=1 \mod{2}$. Hence, $q=3\mod{4}$. 
\qed
\end{proof}

The condition $q=3\mod{4}$ can
also be satisfied efficiently by the following lemma,

\begin{lemma}\label{prime}
For all $q,n\in\Z^+$ and primes $p$, the condition $q=p^n =3 \mod{4}$ holds true 
if and only if $p=3\mod{4}$ and $n$ is odd. 
\end{lemma} 

\begin{proof}
The condition is, 
\begin{align*}
q=p^n &=3\mod{4}\\
p^{n \mod{\phi(4})} &= 3\mod{4}\\
p^{n \mod{2}} &= 3\mod{4}.
\end{align*}

The left hand side yields $3\mod{4}$ if and only if $p=3\mod{4}$ and $n$ is odd.
\qed
\end{proof}

\begin{lemma}\label{sroot}
In a finite field of odd order $q=3\mod{4}$, for all $a\in\QR^*_q$,
square root of $a$ is computable in $\O(\log{q})$ squaring as
$a^{\frac{q+1}{4}}=a^\frac{1}{2}\in\F^*_q$.  
\end{lemma}

\begin{proof}
The exponent $\frac{q+1}{4}\in\Z^+$ if and only if $q=3\mod{4}$. 
By lemma~\ref{order}, $a^{\frac{q+1}{2}}=a\in\QR_q$. So,
$a^{\frac{q+1}{4}}=a^\frac{1}{2}\in\F^*_q$. It takes $(\lceil\log{(q+1)}\rceil-2)$ 
squaring in $\F^*_q$ to compute $a^{\frac{q+1}{4}}$.  
\qed
\end{proof}

\begin{assumption}\label{bound}
{\normalfont \textbf{(Lower bound on finding square roots in $\F^*_q$)}}
In a finite field of odd order $q$, for all $a\in\QR_q$, computing the square 
root of $a$ takes $\Omega(\log{q})$ time.
\end{assumption}

\begin{proof}
We consider the algorithm by Doliskani and Schost to be fastest algorithm to find the
square roots in sufficiently large finite fields~\cite{Doliskani14}. The proposition 
3.2 in cf.~\cite{Doliskani14} shows that their algorithm runs in time
$\O(sM(n)\log{p}+C(n)\log{n})$ for $q=p^n$ where $M(n)=n\log{n}\log{\log{n}}$ and
$C(n)=\O(n^{1.67})$ denote the time for multiplication and composition of polynomials.
For square roots $s$ turns out to be $1$. Therefore, it takes
$\Omega(n\log{p})=\Omega(\log{q})$-time.
\end{proof}

\subsection{Verifiable Delay Function}\label{VDF}
We borrow this formalization from~\cite{Dan2018VDF}.

\begin{definition}
{\normalfont{ \bf (Verifiable Delay Function).}}
A verifiable delay function from domain $\X$ to range $\Y$ is a tuple of algorithms 
$(\setup, \eval, \Verify)$ defined as follows,
\begin{itemize}[label=\textbullet]
\item $\setup(1^\lambda, T) \rightarrow \mathbf{pp}$
 is a randomized algorithm that takes as input a security parameter $\lambda$ 
 and a delay parameter $T$, and produces the public parameters 
 $\pp$. We require $\setup$ to run in $\poly(\lambda,\log{T})$-time.
 
 \item $\eval(\pp, x) \rightarrow (y, \pi)$ takes an input 
 $x\in\mathcal{X}$, and produces an output $y\in\mathcal{Y}$ and a (possibly empty) 
 proof $\pi$. $\eval$ may use random bits to generate the proof 
 $\pi$. For all $\mathbf{pp}$ generated 
 by $\setup(1^\lambda, T)$ and all $x\in\mathcal{X}$, the algorithm 
 $\eval(\mathbf{pp}, x)$ must run in time $T$. 
 
 \item $\Verify(\pp, x, y, \pi) \rightarrow \{0, 1\}$ is a 
 deterministic algorithm that takes an input $x\in\mathcal{X}$, an output $y\in\mathcal{Y}$,
 and a proof $\pi$ (if any), and either accepts $(1)$ or rejects $(0)$. 
 The algorithm must run in $\poly(\lambda,\log{T})$ time.
\end{itemize}
\end{definition}

Before we proceed to the security of $\vdf$s we need the precise model of parallel
adversaries~\cite{Dan2018VDF}. 
\begin{definition}{\normalfont{(\bf Parallel Adversary)}}\label{paradv} 
A parallel adversary $\adv=(\adv_0,\adv_1)$ is a pair of non-uniform 
randomized algorithms $\adv_0$ with total running time $\poly(\lambda,T)$, 
and $\adv_1$ which runs in parallel time $\sigma(T)$ on at 
most $\poly(\lambda,T)$ number of processors.
\end{definition}
Here, $\adv_0$ is a preprocessing algorithm that precomputes some
$\st$ based only on the public parameters, and $\adv_1$ exploits
this additional knowledge to solve in parallel running time $\sigma$ on 
$\poly(\lambda,T)$ processors.

Three essential security properties of $\vdf$s are now described.

\begin{definition}{\normalfont{(\bf Correctness)}}\label{def: Correctness} 
A $\vdf$ is correct if for all $\lambda, T\in\Z^+$, and 
$x\in\mathcal{X}$, we have,
\[
\Pr\left[
\begin{array}{l}
\Verify(\mathbf{pp},x,y,\pi)=1
\end{array}
\Biggm| \begin{array}{l}
\mathbf{pp}\leftarrow\setup(1^\lambda,T)\\
x\xleftarrow{\$} \mathcal{X}\\
(y,\pi)\leftarrow \eval(\mathbf{pp},x)
\end{array}
\right]
= 1.
\]
\end{definition}

\begin{definition}{\normalfont{\bf(Soundness)}}\label{def: Soundness} 
A $\vdf$ is computationally sound if for all non-uniform algorithms $\adv$ 
that run in time $\poly(\lambda,T)$,
we have,
\[
\Pr\left[
\begin{array}{l}
y\ne\eval(\mathbf{pp},x)\\
\Verify(\mathbf{pp},x,y,\pi)=1
\end{array}
\Biggm| \begin{array}{l}
\mathbf{pp}\leftarrow\setup(1^\lambda,T)\\
(x,y,\pi)\leftarrow\adv(1^\lambda,1^T,\mathbf{pp})
\end{array}
\right] \le \negl(\lambda).
\]
\end{definition}

Further, a $\vdf$ is called statistically sound when all adversaries 
(even computationally unbounded) have at most $\negl(\lambda)$ advantage.
Even further, it is called perfectly sound if we want this probability to be $0$ 
against all adversaries. Hence, perfect soundness implies statistical soundness which
implies computational soundness but not the reverses. 

\begin{definition}{\normalfont{\bf (Sequentiality)}}\label{def: Sequentiality}
A $\vdf$ is $(\Gamma,\sigma)$-sequential if there exists no
pair of randomized algorithms $\adv_0$ with total running time
$\poly(\lambda,T)$ and $\adv_1$ which runs
in parallel time $\sigma(T)<T$ on at most $\Gamma$ processors, such that
\[
\Pr\left[
\begin{array}{l}
y=\eval(\mathbf{pp},x)
\end{array}
\Biggm| \begin{array}{l}
\mathbf{pp}\leftarrow\setup(1^\lambda,T)\\
\st\leftarrow\adv_0(1^\lambda,T,\mathbf{pp})\\
x\xleftarrow{\$}\mathcal{X}\\
y\leftarrow\adv_1(\st,x)
\end{array}
\right]
\le \negl(\lambda).
\]
\end{definition}
$\sigma(T)=T$ is impossible as $\eval$ runs in time $T$. In all practical applications 
it suffices to attain $\sigma(T)= (1-\epsilon) T$ for sufficiently small $\epsilon$ 
while an almost-perfect $\vdf$ would achieve $\sigma(T)=T-o(T)$~\cite{Dan2018VDF}.

\section{Single Squaring Verifiable Delay Function}\label{LPoSW}
%
%

As before, $\lambda\in\mathbb{Z}^+$ denotes the security parameter, 
$T\in 2^{o(\lambda)}$ denotes the delay parameter.
The three algorithms that specify our $\vdf$ are,

\subsection{The $\setup(1^\lambda,T)$ Algorithm}\label{Gen}
This algorithm outputs the public parameters
$\mathbf{pp}=\langle \F_q, \RO \rangle$ having
the following meanings.
\begin{enumerate}

\item $\F_q$ is a finite field of order $q\ge 2^\lambda$ such that $q=3 \mod{4}$.
We denote $\F^*_q$ as the multiplicative group of $\F_q$. It is a cyclic group of order
$q-1$.

\item 

We take $\RO:\X\rightarrow\QR^*_q$ to 
be a random oracle that maps an input statement $x\in\mathcal{X}$ to an element in $\QR^*_q$.
\end{enumerate}

  None of the public parameters needs to be computed. The cost of $\setup$ has been
analyzed in Sect.~\ref{Efficiency}.

\subsection{The $\eval$ Algorithm}
The prover $\prv$ needs to compute one of the square roots of the element $\RO(x)$
from the $\QR^*_q$. In particular, $\prv$ executes,

\begin{algorithm}
\caption{$\eval(\pp,x,T)\rightarrow y$}
\begin{algorithmic}[1]
 \State compute $g:=\RO(x)\in\QR^*_q$.
 \State compute $y:=g^{\frac{q+1}{4}}\in\F^*_q$. 
 \Comment{By lemma~\ref{sroot}, $y$ is a square root of $s$ over $\F^*_q$ .}\\

 \Return $y$
\end{algorithmic}
\end{algorithm}
 $\prv$ announces the triple $( x, T, y)$. 

\subsection{The $\Verify$ Algorithm}
The verifier $\vrf$ only checks if $y^2=\RO(x)$. 
So, $\vrf$ runs the following algorithm,

\begin{algorithm}
\caption{$\Verify(\pp,x,T, y,\bot)\rightarrow \{0,1\}$}
\begin{algorithmic}[1]
 \State compute $g:=\RO(x)\in\QR^*_q$.

 \If {$y^2 = g$}
	\State \Return $1$
 \Else	
	\State \Return $0$
 \EndIf	
\end{algorithmic}
\end{algorithm}

Verification needs no proof, so $\pi=\bot$.


\section{Security Analysis}\label{properties}
Here we show that the proposed $\vdf$ is correct, sound and sequential.

\subsection{Correctness}
The verifier should always accept a valid triple $(x,y,T)$.

\begin{theorem}
 The single squaring $\vdf$ is correct.
\end{theorem}
\begin{proof}
 Since $\RO$ is a random oracle, $g\in\QR^*_q$ is uniquely determined
 by the challenge $x\in\X$. If the prover $\prv$ finds the square 
 root of $\RO(x)$ then the verifier $\vrf$ accepts the tuple $(\pp,x,y,T)$ as 
 $y^2=\RO(x)$ by lemma~\ref{sroot}.
\qed 
\end{proof}

\subsection{Soundness}
 Using an invalid proof no adversary $\adv$ having $\poly(T)$ processors should convince the verifier 
 with non-negligible probability.

\begin{definition}{\normalfont \textbf{(Square Root Game $\sqrt{\mathcal{G}}_\F$)}}\label{root} 
Let $\adv$ be a party playing the game. 
The square root finding game $\sqrt{\mathcal{G}}_{\F}$ goes as follows: 
\begin{enumerate}
 \item $\adv$ is given a composite integer $\F_q$.
 \item $\adv$ is given an element $z\xleftarrow{\$}\QR^*_q$ chosen uniformly at random.
 \item Observing $z$, $\adv$ outputs an element $w\in\F^*_q$.
\end{enumerate}
The player $\adv$ wins the game $\sqrt{\mathcal{G}}_\F$ if $w^2=z\in\QR^*_q$.
\end{definition}

%
%
By lemma~\ref{order}, there exists exactly two square roots of each $z\in\QR^*_q$. 

However, we need another version of the game $\sqrt{\mathcal{G}}_\F$ 
in order to reduce the soundness-breaking game.
Instead of sampling $z$ uniformly at random, in this 
version we compute $\RO : \X \rightarrow z\in\QR^*_q$.

\begin{definition}
{\normalfont \textbf{(Square Root Oracle Game $\sqrt{\mathcal{G}}^\RO_\F$).}}\label{funcroot} 
Let $\adv$ be a party playing the game. Suppose, 
$\RO:\{0,1\}^* \rightarrow \QR^*_q$ is random oracle that  
always maps any element from its domain to a fixed
element chosen uniformly at random from its range,
with the probability $\frac{1}{|\QR^*_q|}$.
The Square Root Function Game $\sqrt{\mathcal{G}}^\RO_\F$ goes as follows: 
\begin{enumerate}
 \item $\adv$ is given the output of the 
 $\setup(1^\lambda,T)\rightarrow \langle \RO, \F_q \rangle $.
 \item An element $x\xleftarrow{\$}\X$ is chosen uniformly at random.
 \item $\adv$ is given the element $\RO(x)\rightarrow z\in\QR^*_q$.
 \item Observing $z$, $\adv$ outputs an element
$w\in\F^*_q$.
\end{enumerate}
The player $\adv$ wins the game $\sqrt{\mathcal{G}}_\F^\RO$ if $w^2=z\in\QR^*_q$.
\end{definition}

Under the assumption that $\RO$ samples elements from its range uniformly at random
the distributions of $z$ in 
 both the games $\sqrt{\mathcal{G}}_\F$ and $\sqrt{\mathcal{G}}^\RO_\F$ 
should be indistinguishable for $\adv$. Therefore, $\adv$ has exactly two solutions to 
win the game $\sqrt{\mathcal{G}}^\RO_\F$. 
%
%
%
%
%
%
%

\begin{theorem}{\normalfont \textbf{(Soundness).}}\label{thm: soundness}
 Suppose $\adv$ be a player who breaks the soundness of the single squaring $\vdf$
 with probability $p_{win}$. Then there is a player
 $\adf$ who wins the square root oracle game $\sqrt{\mathcal{G}}^\RO_\F$ 
 with the probability $p_{win}$. 
\end{theorem}
\begin{proof}
 Run the $\setup(1^\lambda,T)$ to obtain $\langle \RO,\F_q \rangle$. 
 The player $\adf$ calls $\adv$ on the input statement $x\xleftarrow{\$}\mathcal{X}$.
 Suppose $\adv$ outputs $y\ne\eval(\pp,x,y,T)\in \F^*_q$. 
 As $\adv$ breaks the soundness of this $\vdf$, $\Verify(\pp,x,y,T)=1$, so 
 $y^{2}=\RO(x)\in\F^*_q$ 
 with the probability $p_{win}$. So $\adf$ computes $y^2$ from $y$
 and outputs it to win the game $\sqrt{\mathcal{G}}^\RO_T$ 
 with probability $p_{win}$.
\qed 
\end{proof}

%

\subsection{Sequentiality}\label{Sequentiality}
The sequentiality analysis of our $\vdf$ scheme is based on 
the sequentiality of time-lock puzzle~\cite{Rivest1996Time}.

\begin{definition}{\normalfont(\textbf {Generalized Time-lock Game $\mathcal{G}^{T}_{\G})$}}
Let $\adv=(\adv_0,\adv_1)$ be a party playing the game. For the security parameter
$\lambda\in\Z^+$ and the delay parameter $T=T(\lambda)\in\Z^+$, the time-lock game 
$\mathcal{G}^{T}_{\G}$ goes as follows: 
\begin{enumerate}
 \item $\adv_0$ is given a group $\G$ of order $|\G|\ge 2^\lambda$.
 \item $\adv_0$ computes some information $\st\in \{0,1\}^*$ on the input $\G$.
 \item $\adv_1$ is given the information $\st$ and an element $g\xleftarrow{\$}\G$ chosen uniformly at random.
 \item Observing $g$ and $\st$, $\adv_1$ outputs an element $h\in\G$ in time $<T$.
\end{enumerate}
The player $\adv=(\adv_0,\adv_1)$ wins the game $\mathcal{G}^{T}_{\G}$ if
$h=g^{2^T}\in{\G}$.
\end{definition}

\subsubsection{Hardness of Time-lock puzzle}
Here we describe two cases where time-lock puzzles are believed to be sequentially hard. 
The first one is celebrated as the time-lock puzzle in the group of unknown order~\cite{Rivest1996Time}.
The same idea has been conjectured for the class group of imaginary quadratic number
fields $Cl(d)$ in~\cite{Wesolowski2019Efficient}. We bring into a similar sequentially hard 
assumption but in the group of known order.

\begin{assumption}{\normalfont \textbf{(Time-lock puzzle in the group of unknown order).}} 
For all algorithms $\adv=(\adv_0,\adv_1)$ running in time $<T$ on $\poly(T)$ number of processors, 
the probability that $\adv$ wins the time-lock game $\mathcal{G}^{T}_{\G}$ is at most $\negl(\lambda)$
when the order of the group $\G$ is unknown~\cite{Rivest1996Time}. Formally,
$$\Pr[\adv \text{ wins } \mathcal{G}^{T}_{\G}\mid|\G| \text{ is unknown}] \le \negl(\lambda).$$
\end{assumption}

%
%
However, knowing the order does not necessarily ease the computation sequentiality. 
It is because, till date, the only faster algorithm that we know to compute $g^e$ 
for any $e \in \Z^+$ is to reduce it to $g^{e \mod |\G|}$ assuming that the group $\G$
is not efficiently solvable into smaller groups. Therefore, when $e<|\G|$ there are only two
options we are left with. If $e\le  e-|\G|)$ then compute $g^e$, otherwise,
compute $(g^{-1})^{|\G|-e}$.
Clearly, the knowledge of $|\G|$ reduces the number of required squaring only if $e>
|\G|/2$ but keeps the same when $e \le |\G|/2$. Therefore, time-lock
puzzle holds true even when the order $|\G|$ is known and the exponent is $e \le |\G|/2$.

\begin{assumption}
{\normalfont \textbf{(Time-lock puzzle in the group of known order).}}\label{known} 
For all algorithms $\adv=(\adv_0,\adv_1)$ running in time $<T$ on $\poly(T)$ number 
of processors, the probability that $\adv$ wins the time-lock game $\mathcal{G}^{T}_{\G}$ 
is at most $\negl(\lambda)$ when the order of the group $\G$ is known to be at least $2^{T+1}$. 
Formally, $$\Pr[\adv \text{ wins } \mathcal{G}^{T}_{\G}\mid|\G| \ge 2^{T+1}] \le \negl(\lambda).$$
\end{assumption}

\begin{definition}
{\normalfont\textbf{(Time-lock Oracle Game $\mathcal{G}^{\RO^{T}}_{\G})$}}\label{time} 
Let $\adv=(\adv_0,\adv_1)$ be a party playing the game. For the security parameter
$\lambda\in\Z^+$ and the delay parameter $T=T(\lambda)\in \Z^+$, the time-lock oracle game 
$\mathcal{G}^{\RO^{T}}_{\G}$ goes as follows: 
\begin{enumerate}
 \item $\adv_0$ is given the output of the $\setup(1^\lambda,T)\rightarrow \langle \RO,\G \rangle $.
 
 \item $\adv_0$ computes some information $\st\in \{0,1\}^*$ on the input $\G$ and $\RO$.
 \item An element $x \xleftarrow{\$} \X$ is chosen uniformly at random.
 \item $\adv_1$ is given the information $\st$ and an element $g\leftarrow \RO(x)$.
 \item Observing $g$ and $\st$, $\adv_1$ outputs an element $h\in\G$ in time $<T$.
\end{enumerate}
The player $\adv=(\adv_0,\adv_1)$ wins the game $\mathcal{G}^{T}_{\G}$ if
$h=g^{2^T}\in\G$.
\end{definition}

Under the assumption that the random oracle $\RO$ samples strings from its range uniformly at random 
the view of the distributions of $g$ in both the games $\mathcal{G}^{{T}}_{\G}$ and
$\mathcal{G}^{\RO^{T}}_{\G}$ are identical to $\adv$. Thus we infer that,

\begin{assumption} {\normalfont {\textbf{(Time-lock Oracle
Assumption).}}}\label{timelock}
 For all algorithms $\adv=(\adv_0,\adv_1)$ running in time $<T$ on $\poly(T)$ number
 of processors wins the time-lock oracle game $\mathcal{G}^{\RO^{T}}_{\G}$ with 
 the probability at most negligible in the security parameter $\lambda$. Mathematically,
\[
\Pr\left[
\begin{array}{l}
g^{2^T}=h
\end{array}
\Biggm| \begin{array}{l}
(\RO,\G)\leftarrow\setup(1^\lambda,T)\\
\st \leftarrow \adv_0(1^\lambda,T,\RO,\G)\\
x \xleftarrow{\$} \X \\
g \leftarrow \RO(x)\\
h \leftarrow\adv_1(g,\st)
\end{array}
\right] \le \negl(\lambda).
\]
\end{assumption}

The game $\sqrt{\mathcal{G}}^\RO_\F$ (Def.~\ref{funcroot}) is an example
of time-lock puzzle in the group of known order as $|\F^*_q|=q-1$ 
and $w=z^{\frac{q+1}{4}}\in\F^*_q$. By assumption~\ref{bound}, 
computing $w \in\G$ needs $\Omega(\log{q})$ squaring for any 
$z\in\QR^*_q$. Therefore, the game $\sqrt{\mathcal{G}}^\RO_\F$ 
(Def.~\ref{funcroot}) satisfies assumption~\ref{timelock} for $\G=\F^*_q$ and $T=\log{q}$.

\begin{theorem}\label{thm: sequentiality} {\normalfont \textbf{(Sequentiality).}}
Suppose $\adv$ be a player who breaks the sequentiality of the single squaring $\vdf$
with probability $p_{win}$. Then there is a player $\adf$ who wins the square root oracle 
game $\sqrt{\mathcal{G}}^\RO_\F$ (Def.~\ref{funcroot}) with probability $p_{win}$.
\end{theorem}
\begin{proof}
 Run the $\setup(1^\lambda,T=\log{q})$ to obtain $\langle \RO,\F_q \rangle$. 
 $\adf$ calls $\adv$ on the input statement $x\xleftarrow{\$}\mathcal{X}$.
 Against $x$, $\adv$ outputs $y\in\F^*_q$ in time $< \log{q}$. 
 As $\adv$ breaks the sequentiality of this $\vdf$, $y^2=\RO(x) \in\F^*_q$ 
 with the probability $p_{win}$. So $\adf$ outputs $y$ and wins the game 
 $\sqrt{\mathcal{G}}^\RO_\F$ with probability $p_{win}$. 
\qed 
\end{proof}

In the next section, we show that our $\vdf$ is $(\poly(\lambda),(1-\epsilon)T)$-sequential. 

\section{Efficiency Analysis}\label{Efficiency}
Here we discuss the efficiencies of both the prover $\prv$ and the verifier 
$\vrf$ in terms of number of modulo squaring and the memory requirement.

\subsubsection{Cost of $\setup$}
The cost of constructing $\RO$ is
equal to that of constructing another random oracle $\RO':\X \rightarrow
\{\F^*_q\setminus{1}\}$ plus a single squaring over $\F^*_q$ because 
$\RO(x)\stackrel{def}{=}\RO'(x)^2$. The random oracles of type $\RO'$ are 
standard in this domain of research~\cite{Wesolowski2019Efficient,Pietrzak2019Simple}.

By lemma~\ref{prime} we need a prime $p=3\mod{4}$ and an odd $n$ for $q=3\mod{4}$.
Any prime selected uniformly at random satisfies the condition $p=3\mod{4}$ with
probability at least $1/2$. The dimension $n$ is chosen to satisfy $p^n\ge 2^\lambda$. 
The enumeration of the irreducible polynomial $f(x)$ need not be the part 
of $\setup$ as it may be sampled from a precomputed list. However, we recommend to use 
computer algebra systems like SAGEMATH in order to reliably sample $f(x)$. 

 \subsubsection{Proof Size} 
 The proof size is essentially zero as the proof $\pi=\bot$ is empty. 
 The output $y$ is an element in the group $\F^*_q$. So the size of 
 the output is $\log{q}$-bits.
 
 \subsubsection{Prover's Efficiency} As already mentioned in Theorem~\ref{thm: sequentiality}, 
 the prover $\prv$ needs at least $T=\log{q}$ sequential time in order to compute 
 $\RO(x)^{\frac{q+1}{4}}$. 
 The prover needs no proof. Thus the prover's effort is not only at least
 $T=\Theta(\log{q})$. We use the $\Theta$ notation to mean $(1-\epsilon)T$ for some
 $\epsilon \rightarrow 0$. We keep this provision because $\prv$ may speed up the field
 multiplications using $\poly(T)$ processors but not in time $<T-o(T)$.

\subsubsection{Verifier's Efficiency}
In this $\vdf$, verification needs only a single squaring over $\F^*_q$. 

\subsubsection{Upper bound on $T$} As the output by the prover needs $\log{q}$-bits,
we need to restrict $\log{q} \in \poly(\lambda)$ at most. Cryptographic protocols that deal
with objects larger than the polynomial-size in their security parameter are impractical.
Hence, for the group $\F^*_q$ the delay parameter $T=\log{q}$ must be bounded by
$\poly(\lambda)$.

We note that handling an element in the group $\F^*_q$ does not imply
the sequential effort by the verifier is proportional to $T=\log{q}$. As a comparison, we never
consider the complexity of operation over the group $\multgroup{N}$ for the $\vdf$s
in~\cite{Wesolowski2019Efficient,Pietrzak2019Simple} in table~\ref{tab : VDF}.
In fact, the sequentiality of squaring does not depend on the complexity of the
elementary group operation in the underlying group. The sequentiality should hold even
if these elementary operations are of in $\O(1)$-time.  


\subsection{Performance Comparison}\label{compare}
In this section, we demonstrate a few advantages of our $\vdf$.
\begin{description}
 
 \item [Lower bound on prover's parallelism]  
 The $\vdf$ based on injective rational maps demands $\poly(\lambda,T)$-parallelism 
 to compute the $\eval$ in time $T$~\cite{Dan2018VDF}. 
 In our case, no parallelism is required to compute the $\vdf$ in time $T$.

 \item [Proof size] 
 The proofs in the Wesolowski's and Pietrzak's $\vdf$s consume one and $\log{T}$ group 
 elements in their underlying group $\G$ (i.e.,$\multgroup{N}$ and $Cl(d)$)~\cite{Wesolowski2019Efficient,Pietrzak2019Simple}. 
 In order to communicate the proof and to verify efficiently $\log{|\G|} \in \poly(\lambda)$. 
 Our $\vdf$ requires no proof. 
 
 
 \item [Efficiency of $\setup$] The $\setup$ in the isogeny-based $\vdf$ may turn out to be as slow as 
  the $\eval$ itself~\cite{Feo2019Isogenie}. On the other hand, the $\setup$ in our $\vdf$ is
 as efficient as that in all other time-lock based $\vdf$s~\cite{Wesolowski2019Efficient,Pietrzak2019Simple}.

  \item [$\Omega(\lambda)$-Verifiability]
  All the existing $\vdf$s except~\cite{Dan2018VDF} at least $\lambda$ sequential effort for verification.
  Sect.~\ref{minseq} describes the concrete bound of the time-lock puzzle based $\vdf$s.
  On the contrary, the most important advantage of our $\vdf$ is that the verification
 requires only one squaring over $\F^*_q$.


  \item [Largeness of $T$] A limitation of our $\vdf$ is that it works only when
  $T=\log{q}\in \poly(\lambda)$ as the output size is dominated by $T$, 
  while all the above-mentioned $\vdf$s allow $T \in 2^{o(\lambda)}$.
  
\end{description}

In Table~\ref{tab : VDF}, we summarize the above description. 
\begin{table*}[h]
\caption{Comparison among the existing $\vdf$s. $T$ is the delay parameter, $\lambda$ is 
the security parameter and $\Gamma$ is the number of processors. All the quantities may be
subjected to $\mathcal{O}$-notation, if needed.}
\label{tab : VDF}
 \centering
 \begin{tabular}{|l@{\quad}|r@{\quad}|r@{\quad}|r@{\quad}|r@{\quad}|r@{\quad}|r@{\quad}}
    \hline
        $\vdf$s & $\eval$ & $\eval$ &  $\Verify$ & $\setup$ & Proof   \\
(by authors)&  Sequential   & Parallel     &       &         &  size   \\
    \hline
    
        
    Boneh et al.~\cite{Dan2018VDF}         & $T^2$ & $>T-o(T)$  &  $\log{T}$  & $\log{T}$ & $\textendash$  \\
    [0.3 em] \hline
    
    Wesolowski~\cite{Wesolowski2019Efficient} & $(1+\frac{2}{\log{T}})T$
& $(1+\frac{2}{\Gamma\log{T}})T$  &  $\lambda\log{\lambda}$  & $\lambda^{3}$ & $\lambda^{3}$ \\
    [0.3 em] \hline
    
    Pietrzak~\cite{Pietrzak2019Simple}        & $(1+\frac{2}{\sqrt{T}})T$
& $(1+\frac{2}{\Gamma\sqrt{T}})T$  &  $\lambda\log{T}$  & $\lambda^{3}$ & $\log{T}$ \\
    [0.3 em] \hline

    Ephraim et al.~\cite{Ephraim20VDF}        & $(1+\frac{2}{\sqrt{T}})T$
& $(1+\frac{2}{\Gamma\sqrt{T}})T$  &  $\lambda$  & $\lambda^{3}$ & $\log{T}$ \\
    [0.3 em] \hline 

    Feo et al.~\cite{Feo2019Isogenie}         & $T$   & $T$  &  $\lambda$  & $T\log{\lambda}$ & $\textendash$ \\
    [0.3 em] \hline 
    
   \textbf{Our work}                         & $T$  & $(1-\epsilon)T$  &  $1$  & $\lambda^3$ & $\textendash$ \\
    [0.3 em]\hline
 \end{tabular}
\end{table*}

\section{Conclusions}\label{dis}

This chapter derives the first $\vdf$ that verifies its computation with a single
squaring. It has been derived from a new sequential assumption namely the time-lock
puzzle in the group of known order. At the same time, this is a $\vdf$ from the
complexity class of all deterministic computation as we have shown the delay parameter 
$T$ is in $\poly(\lambda)$.


\end{document}